\documentclass[%
 reprint,
 amsmath,amssymb,
 aps,
 prb
]{revtex4-2}

\usepackage{bbding}
\usepackage{graphicx}
\usepackage{dcolumn}
\usepackage{bm}
\usepackage{amsmath}      
\usepackage{amssymb}      
\usepackage{amsfonts}     
\usepackage{amsthm}       
\usepackage{amstext}
\usepackage{amsbsy}
\usepackage{mathtools} 
\usepackage{mathtools}
\usepackage{color}
\usepackage{xcolor}       
\usepackage{wasysym}
\usepackage[pdf]{pstricks}

\usepackage[colorlinks=true,linkcolor=blue,urlcolor=blue,filecolor=magenta,citecolor=blue,breaklinks]{hyperref}
\hypersetup{
    pdftitle={Overleaf Example},
    pdfpagemode=FullScreen
}
\usepackage[capitalise]{cleveref}
\usepackage{tikz}
\usetikzlibrary{shapes.geometric, arrows}

\newtheorem{lemma}{Lemma}

\newtheorem{nres}{Numerical Result}

\begin{document}


\title{Thermal Order by Disorder in Resonating-Valence Bond States\\ on the Checkerboard Lattice}

\author{Giorgi Gogaberishvili}
\author{Nika Kurdadze}
\author{Kirill Shtengel}

\affiliation{
 Department of Physics and Astronomy, University of California Riverside, Riverside, CA 92521, USA 
}%

\date{\today}

\begin{abstract}
We derive a local spin-1/2 Hamiltonian with a resonating valence bond ground state on the checkerboard lattice. The state is characterized by the exponential decay of singlet--singlet correlations,  whereas dimer--dimer correlations decay with a power law in the corresponding Quantum Dimer model. This observation leads to a novel mechanism for thermal Order by Disorder whereby thermal decoherence suppresses destructive quantum interference between different contributions to the correlations in the ground state and results in a qualitatively different, quasi-long-range ordered mixed state.
\end{abstract}

\maketitle


\subsubsection{\label{sec:intro} Introduction}
Resonating valence bond (RVB) states introduced half a century ago \cite{Anderson1973} laid the foundation for our understanding of spin liquids. While short-ranged RVB states are conjectured to be ground states of frustrated spin Hamiltonians, tractable examples of such Hamiltonians remain in short supply~\cite{Cano2010}. In their absence, a lot of studies addressed the properties of Quantum Dimer (QD) models \cite{Rokhsar1988,Moessner2001a,Moessner2008a}. The key difference is that quantum dimers are treated as an orthogonal basis of quantum states, whereas different singlet states, which they aim to mimic, are not mutually orthogonal. This subtle distinction leads to dramatic differences in the properties of short-range RVB and QD states. Specifically, even on the square lattice where both dimer--dimer and monomer--monomer correlations decay with a power law, the spin--spin correlations decay exponentially \cite{Liang1988}. Nevertheless, singlet--singlet correlations still decay with a power law~\cite{Tang2011a}, albeit slower than dimer-dimer correlation in the corresponding Rokhsar--Kivelson (RK) model \cite{Rokhsar1988}. 
This distinction has been elucidated by Damle et al.~\cite{Damle2012}.

In the presence of geometric frustration, this behavior changes dramatically: both dimer--dimer correlations in the RK-type QD models and singlet--singlet correlations in the corresponding RVB states decay exponentially~\cite{Moessner2001a,Yang2012,Wildeboer2012}. 
From that point of view, a surprising development due to Batista and Trugman \cite{Batista2004} was the Klein-type spin-1/2 model on the checkerboard lattice, where the ground state manifold is in one-to-one correspondence with the states of the six-vertex model (at least with the periodic boundary conditions)~\cite{Nussinov2007a}. The correlations in the classical ensemble of all dimer coverings corresponding to the spin singlets in the ground state manifold would thus be described by the correlations in the six-vertex model at its isotropic point (with equal weight for all vertices) and thus decay as a power law ($\propto 1/r^2$) even though the underlying lattice is not bipartite! This behavior is a consequence of an additional restriction on the dimer configurations: in addition to one dimer per site, only one dimer per plaquette is allowed. However, two factors complicate making statements about single-singlet correlations in this model. Firstly, singlets on this lattice cannot be consistently oriented in such a way that would result in positive amplitudes for all terms in the corresponding RVB state. Secondly, no Hamiltonian that would have such an RVB state as its ground state was known; the additional Heisenberg terms analyzed in  Ref.~\cite{Nussinov2007a} were found to result in a valence-bond solid (VBS) instead.

In this manuscript, we remedy the latter concern by constructing a local SU(2) invariant spin-1/2 Hamiltonian that, combined with the original Batista--Trugman Hamiltonian, results in an equal-weight superposition of all such singlet coverings. Furthermore, this superposition can be represented as a fermionic RVB state~\cite{Yang2012} and thus efficiently simulated. We find numerically that the singlet--singlet correlations fall off exponentially in this state.

Furthermore, we identify a novel mechanism of thermal Order By Disorder, originating from destroying destructive interference between different singlet states rather than by entropic selection of a preferred sector.
Specifically, thermal decoherence results in suppressing the off-diagonal matrix elements of the density matrix which in turn drives the quantum-disordered RVB ground state with exponentially decaying singlet-singlet correlations into an incoherent collection of valence bond states, the ground states of the Klein-type Hamiltonian~\cite{Batista2004}. In this mixed state, singlet-singlet correlations exhibit the same $1/r^2$ decay as QD dimer-dimer correlations, thus manifesting a quasi-long-range order.

\subsubsection{\label{sec:RVB} RVB state on the checkerboard lattice}

Our starting point is the Klein-type Hamiltonian on the checkerboard lattice~\cite{Batista2004}:
\begin{equation}
    \label{eq:Klein}
   \hat H_\text{K} = K \sum_p \mathbf{S}_{p}^{\,2} \left(\mathbf{S}_{p}^{\,2}-2\right)\propto \mathbb{P}_2\left(\mathbf{S}_{p}\right)
\end{equation}
where $\mathbf{S}_{p}$ denotes the total spin on crossed plaquette $p$ (see Fig.~\ref{fig:numbering}). For $K>0$, the Hamiltonian projects out spin states corresponding to the total plaquette spin of $S_p=0$ or $S_p=1$, thus making them zero-energy ground states. Consequently,  any state containing at least one singlet per plaquette is a ground state. While it is possible to have two singlets in one plaquette, the number of such plaquettes in a ground state scales with the system's perimeter in the case of free boundary conditions, and they are outright impossible for the periodic boundary conditions.  This is a consequence of mapping plaquettes with one singlet onto the vertices of the six-vertex model~\cite{Nussinov2007a} -- see Fig.~\ref{fig:six_vert}; in this picture a two-singlet plaquette becomes a ``sink'' for arrows in the eight-vertex model while the corresponding sources would correspond to singlet-free plaquettes, which cost energy.

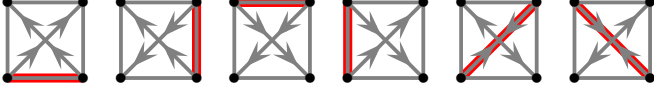
\begin{figure} [h]
    \centering
\begin{pspicture}(-0.1,0)(9,1)
\psline[linewidth=0.12,linecolor=red](0,0)(1,0)
\psline[linewidth=0.05,linecolor=gray](0,0)(1,1)
\psline[linewidth=0.05,linecolor=gray](0,0)(0,1)
\psline[linewidth=0.05,linecolor=gray](0,0)(1,0)
\psline[linewidth=0.05,linecolor=gray](1,0)(1,1)
\psline[linewidth=0.05,linecolor=gray](0,1)(1,0)
\psline[linewidth=0.05,linecolor=gray](0,1)(1,1)
\psline[linewidth=0.1,linecolor=gray]{->}(0.35,0.35)(0.45,0.45)

\psline[linewidth=0.1,linecolor=gray]{-<}(0.7,0.7)(0.6,0.6)
\psline[linewidth=0.1,linecolor=gray]{->}(0.65,0.35)(0.55,0.45)
\psline[linewidth=0.1,linecolor=gray]{-<}(0.3,0.7)(0.4,0.6)
\psdots[dotscale=1](0,0)(1,1)(0,1)(1,0)
%
\psline[linewidth=0.12,linecolor=red](2.5,0)(2.5,1)
\psline[linewidth=0.05,linecolor=gray](1.5,0)(2.5,1)
\psline[linewidth=0.05,linecolor=gray](1.5,0)(1.5,1)
\psline[linewidth=0.05,linecolor=gray](1.5,0)(2.5,0)
\psline[linewidth=0.05,linecolor=gray](2.5,0)(2.5,1)
\psline[linewidth=0.05,linecolor=gray](1.5,1)(2.5,0)
\psline[linewidth=0.05,linecolor=gray](1.5,1)(2.5,1)
\psline[linewidth=0.1,linecolor=gray]{-<}(1.8,0.3)(1.9,0.4)
\psline[linewidth=0.1,linecolor=gray]{->}(2.15,0.65)(2.05,0.55)
\psline[linewidth=0.1,linecolor=gray]{->}(2.15,0.35)(2.05,0.45)
\psline[linewidth=0.1,linecolor=gray]{-<}(1.8,0.7)(1.9,0.6)
\psdots[dotscale=1](1.5,0)(2.5,1)(1.5,1)(2.5,0)
%
\psline[linewidth=0.12,linecolor=red](3,1)(4,1)
\psline[linewidth=0.05,linecolor=gray](3,0)(4,1)
\psline[linewidth=0.05,linecolor=gray](3,0)(3,1)
\psline[linewidth=0.05,linecolor=gray](3,0)(4,0)
\psline[linewidth=0.05,linecolor=gray](4,0)(4,1)
\psline[linewidth=0.05,linecolor=gray](3,1)(4,0)
\psline[linewidth=0.05,linecolor=gray](3,1)(4,1)
\psline[linewidth=0.1,linecolor=gray]{->}(3.35,0.65)(3.45,0.55)
\psline[linewidth=0.1,linecolor=gray]{-<}(3.7,0.3)(3.6,0.4)
\psline[linewidth=0.1,linecolor=gray]{->}(3.65,0.65)(3.55,0.55)
\psline[linewidth=0.1,linecolor=gray]{-<}(3.3,0.3)(3.4,0.4)
\psdots[dotscale=1](3,0)(4,1)(3,1)(4,0)
%
\psline[linewidth=0.12,linecolor=red](4.5,0)(4.5,1)
\psline[linewidth=0.05,linecolor=gray](4.5,0)(5.5,1)
\psline[linewidth=0.05,linecolor=gray](4.5,0)(4.5,1)
\psline[linewidth=0.05,linecolor=gray](4.5,0)(5.5,0)
\psline[linewidth=0.05,linecolor=gray](5.5,0)(5.5,1)
\psline[linewidth=0.05,linecolor=gray](4.5,1)(5.5,0)
\psline[linewidth=0.05,linecolor=gray](4.5,1)(5.5,1)
\psline[linewidth=0.1,linecolor=gray]{-<}(5.2,0.3)(5.1,0.4)
\psline[linewidth=0.1,linecolor=gray]{->}(4.85,0.65)(4.95,0.55)
\psline[linewidth=0.1,linecolor=gray]{->}(4.85,0.35)(4.95,0.45)
\psline[linewidth=0.1,linecolor=gray]{-<}(5.2,0.7)(5.1,0.6)
\psdots[dotscale=1](4.5,0)(5.5,1)(4.5,1)(5.5,0)
%
\psline[linewidth=0.12,linecolor=red](6,0)(7,1)
\psline[linewidth=0.05,linecolor=gray](6,0)(7,1)
\psline[linewidth=0.05,linecolor=gray](6,0)(6,1)
\psline[linewidth=0.05,linecolor=gray](6,0)(7,0)
\psline[linewidth=0.05,linecolor=gray](7,0)(7,1)
\psline[linewidth=0.05,linecolor=gray](6,1)(7,0)
\psline[linewidth=0.05,linecolor=gray](6,1)(7,1)
\psline[linewidth=0.1,linecolor=gray]{->}(6.35,0.35)(6.45,0.45)
\psline[linewidth=0.1,linecolor=gray]{-<}(6.7,0.3)(6.6,0.4)
\psline[linewidth=0.1,linecolor=gray]{->}(6.65,0.65)(6.55,0.55)
\psline[linewidth=0.1,linecolor=gray]{-<}(6.3,0.7)(6.4,0.6)
\psdots[dotscale=1](6,0)(7,1)(6,1)(7,0)
%
\psline[linewidth=0.12,linecolor=red](8.5,0)(7.5,1)
\psline[linewidth=0.05,linecolor=gray](7.5,0)(8.5,1)
\psline[linewidth=0.05,linecolor=gray](7.5,0)(7.5,1)
\psline[linewidth=0.05,linecolor=gray](7.5,0)(8.5,0)
\psline[linewidth=0.05,linecolor=gray](8.5,0)(8.5,1)
\psline[linewidth=0.05,linecolor=gray](7.5,1)(8.5,0)
\psline[linewidth=0.05,linecolor=gray](7.5,1)(8.5,1)
\psline[linewidth=0.1,linecolor=gray]{-<}(7.8,0.3)(7.9,0.4)
\psline[linewidth=0.1,linecolor=gray]{->}(7.85,0.65)(7.95,0.55)
\psline[linewidth=0.1,linecolor=gray]{->}(8.15,0.35)(8.05,0.45)
\psline[linewidth=0.1,linecolor=gray]{-<}(8.2,0.7)(8.1,0.6)
\psdots[dotscale=1](7.5,0)(8.5,1)(7.5,1)(8.5,0)
%
\end{pspicture}
        \caption{Mapping between singlets (shown in red) and six-vertex configurations}
    \label{fig:six_vert}
\end{figure}

Consequently, the Hamiltonian in Eq.~(\ref{eq:Klein}) has exponentially many degenerate ground states; with the number of linearly independent ground states scaling asymptotically as (if not outright equal to) $(4/3)^{3N_p/4}$~\cite{Lieb1967,Baxter}, the number of different six-vertex configurations (see Supplementary Materials). 
Additional interactions will result in mixing between these states and lifting the degeneracy (fully or partially). We are interested in the scenario whereby the resulting ground state is a short-range 
RVB state -- an equal-weight superposition of all singlet coverings within one topological sector. (Topological sectors are defined by the net flux of arrows through the meridian and longitude of the torus or, equivalently, by the global slopes in the height representation.) Since no local interactions can mix different topological sectors, in what follows we will primarily focus on the flat sector, which contains the largest number of states.  

The flat topological sector is ergodic under plaquette flips~\cite{Hermele2004a}. A plaquette flip is defined in the language of the six-vertex model as the arrow reversal on an elementary square around which all four arrows circulate either clockwise or anti-clockwise; an example of such a flip is shown in Fig.~\ref{fig:numbering}. It is therefore suggestive that the effective Hamiltonian for an RVB ground state should have a Rokhsar--Kivelson (RK) form~\cite{Rokhsar1988}:
\begin{equation}
    \label{eq:RK}   
   \hat H_\text{RK}=\left(| \circlearrowright\rangle-|\circlearrowleft\rangle\right)\left(\langle \circlearrowright|-\langle\circlearrowleft|\right),
\end{equation}
where oriented circles represent such flippable plaquettes.
However, due to the non-orthogonality of different singlet configurations, $H_\text{RK}$ does not annihilate non-flippable plaquettes. To alleviate this issue, we introduce local projectors, which project out such plaquettes~\cite{Cano2010}.  One can deduce that the appropriate projector operator annihilating singlets on ``outer'' bonds $(5,6)$, $(7,8)$, $(9,10)$ and $(11,12)$ (see Fig.~\ref{fig:numbering} for the site numbering scheme) is: 
\begin{equation}
   \hat{ \mathbb{P}}_\text{\PlusCenterOpen} =\frac{1}{16}\cdot(1+\hat{P}_{5,6})(1+\hat{P}_{7,8})(1+\hat{P}_{9,10})(1+\hat{P}_{11,12})
   \label{eq:projector}
\end{equation}
  Here, $\hat{P}_{ij}=2\mathbf{S}_i\cdot \mathbf{S}_j+\frac{1}{2}$ denotes a permutation operator that exchanges spins on sites $i$ and $j$. Note that the exclusion of states with singlets on the outer bonds further excludes the possibility of singlets on bonds $(1,2)$, $(2,3)$, $(3,4)$ and $(4,1)$ for a state in the Klein subspace, i.e.\ the subspace of the Hilbert space with exactly one singlet per plaquette. Importantly, acting by the projector operator in Eq.~(\ref{eq:projector}) on a flippable plaquette generates another flippable configuration with the same orientation but outside of the Klein subspace as shown in Fig.~\ref{fig:permutation56}. However, this will not be a problem if we construct an operator that annihilates a superposition of flippable configurations with different orientations irrespective of whether they belong to the Klein subspace.

\begin{figure} [t]
    \centering
    \includegraphics[width=0.85\linewidth]{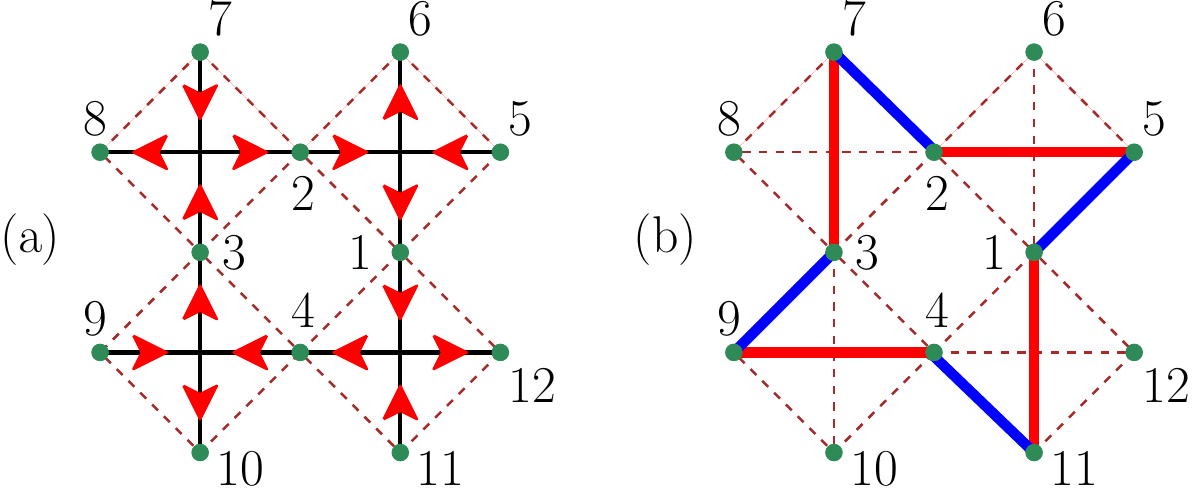}
    \caption{A flippable plaquette in the six-vertex representation (a) and the corresponding pinwheel arrangement of singlets (b) consisting of the original singlets (red) and their rearrangement after the flip (blue).}
    \label{fig:numbering}
\end{figure}

\begin{figure} [b]
    \centering
    \includegraphics[width=0.85\linewidth]{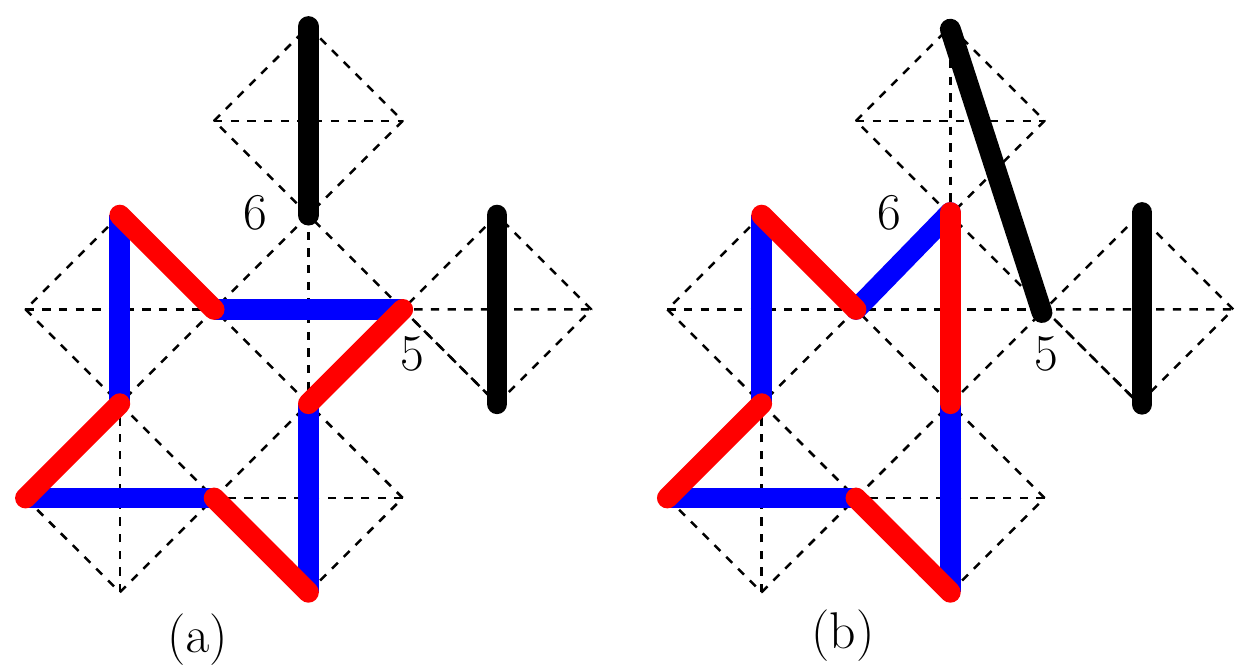}
    \caption{Clockwise (red) and anticlockwise (blue) flippable configurations of singlets (a) before, and (b) after the action of operator $\hat{P}_{56}$. Black singlets are not affected by the flips.}
    \label{fig:permutation56}
\end{figure}

With this in mind, let us focus on the cyclic permutation: $\hat{P}\equiv(1234)=\hat{P}_{3,4}\hat{P}_{2,3}\hat{P}_{1,2}$, which, when applied to a flippable plaquette, effects its flip: $\hat{P}|\circlearrowright\rangle \to |\circlearrowleft\rangle$.  The reason for an arrow instead of an equal sign here is an important subtlety originating from the fact that spin singlets are odd under the permutation of its spins and hence the aforementioned cyclic permutation may also result in an overall sign. So for each pinwheel shape involved in a plaquette flip, we define the pinwheel parity $F_p=f_{ij}f_{jk}f_{kl}f_{lm}f_{mn}f_{nr}f_{rs}f_{st}$ where $(i,j,k,l,m,n,r,s,t)$ are the eight pinwheel corners ordered in the anticlockwise direction and $f_{ij}=1 = - f_{ji}$ in the VB basis containing a singlet orientation $\lvert i\to j\rangle\equiv \left(\uparrow_i \downarrow_j - \downarrow_i \uparrow_j \right)/\sqrt{2}$. E.g., the flippable plaquette in the left pane of Fig.~\ref{fig:numbering} corresponds to the pinwheel $(1,5,2,7,3,9,4,11)$ on the right. With this definition, $\hat{P}|\circlearrowright\rangle = F_p |\circlearrowleft\rangle$. Notably, one can guarantee that \emph{all} such pinwheels have $F_p=+1$ by choosing a basis with singlet orientations such that $F=+1$ for all triangles and empty squares; an example of such orientations is shown in Fig~\ref{fig:singlet_orientation}.  
In such a basis, $\hat{P}|\circlearrowright\rangle=|\circlearrowleft\rangle$ for any flippable plaquette. Since $\hat{P}^4=I$ and consequently $\hat{P}^{-1}=\hat{P}^{3}$, we obtain:
\begin{equation}
\left(\hat{P} + \hat{P}^3 -\hat{P}^2 - \hat{I}\right)
\left(\,|\circlearrowleft\rangle + |\circlearrowright\rangle\,\right)
= 0 .
\end{equation}
\begin{figure} [h]
    \centering
    \includegraphics[width=0.55\linewidth]{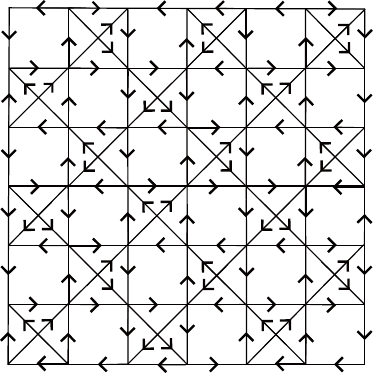}
    \caption{Singlet orientations that result in the positive parity $F_p=+1$ for each pinwheel are indicted by arrows. Note that the parity is positive for each individual triangle and open square.}
    \label{fig:singlet_orientation}
\end{figure}
A  Hermitian operator, which annihilates the superposition of flippable plaquettes and gives positive energy to other states of the Klein subspace (see Supplementary Materials), is given by:
\begin{equation}
\hat{H}' = -\left(\hat{P} + \hat{P}^3 -\hat{P}^2 - \hat{I}\right)
\end{equation}
Note that $\left[\hat{H}',\hat{ \mathbb{P}}_\text{\PlusCenterOpen}\right]=0$ since these operators do not act on the same spins: $\hat{H}'$ involves only inner spins of the flippable configuration in Fig.~\ref{fig:numbering}, whereas $\hat{ \mathbb{P}}_\text{\PlusCenterOpen}$ acts upon the outer spins. Consequently, the equal-weight superposition is indeed annihilated by the $\hat{H}'\hat{\mathbb{P}}_\text{\PlusCenterOpen}$ operator.

Expressing $\hat{P}$ in terms of the participating spin operators, we can write the resulting total local Hamiltonian as:
\begin{multline}
\label{eq:full_Hamiltonian}
    {\hat{H}}=\hat{H}_\text{K}-J \sum_\text{\PlusCenterOpen}\left[\mathbf{S}_1\cdot \mathbf{S}_2+ \mathbf{S}_2\cdot \mathbf{S}_3+\mathbf{S}_3\cdot \mathbf{S}_4+\mathbf{S}_4\cdot \mathbf{S}_1\right.
\\
+4 (\mathbf{S}_1 \cdot \mathbf{S}_2)(\mathbf{S}_3 \cdot \mathbf{S}_4)
+ 4(\mathbf{S}_2 \cdot \mathbf{S}_3)(\mathbf{S}_4 \cdot \mathbf{S}_1)\\
\left.
- 8(\mathbf{S}_1 \cdot \mathbf{S}_3)(\mathbf{S}_2 \cdot\mathbf{S}_4) -\hat{I} \right]  \hat{ \mathbb{P}}_\text{\PlusCenterOpen}
\end{multline}
with $J>0$ and the sum is take over all ``crosses'' consisting of four checkerboard plaquettes, $\hat{I}$ denotes identity matrix. While local and SU(2) symmetric, this Hamiltonian contains twelve-spin interactions.

\subsubsection{Singlet-singlet correlations in the RVB state}
\label{sec:RVB_corr}

In this section, we present our numerical study of singlet--singlet correlations in the RVB state that is the ground state of the Hamiltonian in Eq.~(\ref{eq:full_Hamiltonian}). The RVB state can be written as
\begin{equation}
|\psi_{\mathrm{RVB}}\rangle = \sum_{\mathcal{D}} |\mathcal{D}\rangle,
\qquad
|\mathcal{D}\rangle = \prod_{(ij)\in \mathcal{D}} f_{ij}\, |ij\rangle,
\end{equation}
where the sum is performed over all dimer coverings $\mathcal{D}$ in the Klein subspace. 
This state belongs to a class of so-called bosonic RVB (B-RVB) states~\cite{Yang2012}.
\footnote{Strictly speaking, VB states entering the B-RVB superposition should have an additional factor of $(-1)^{\delta_\mathcal{D}}$ where $\delta_{\mathcal{D}}$ is the number of singlets crossing one another in a given singlet covering  $\mathcal{D}$. In our case,  $\delta_{\mathcal{D}}=0$ for all states of the Klein subspace.}

To avoid the minus-sign problem, we follow the approach developed in Ref.~[\onlinecite{Yang2012}] and map the B-RVB state to the projected-BCS (P-BCS) state:
\begin{multline}
|\psi_{P\text{-BCS}}\rangle
= \sum_{\mathcal{D}}
\prod_{(ij)\in \mathcal{D}}
g_{ij}
\left(
c^{\dagger}_{i\uparrow} c^{\dagger}_{j\downarrow}
- c^{\dagger}_{i\downarrow} c^{\dagger}_{j\uparrow}
\right)
\,|0\rangle .
\end{multline}
where $g_{ij}=g_{ji}$ is the new direction-independent coefficient obtained through the mapping in Ref.~[\onlinecite{Yang2012}]. After the transformation, the singlet-singlet correlation function can be evaluated using variational Monte-Carlo \cite{Edegger2007}.

As shown in Fig.~\ref{fig:correlations},  
 The singlet--singlet
The correlation function exhibits a clear exponential decay with distance~\footnote{In the VMC simulation, the average is taken over all configurations, which might not be linearly independent in general. However, linear independence can be proven when a finite number of singlets are fixed for checkerboard lattice(see Supplementary Materials). We therefore expect that sufficiently distant singlets are effectively unaffected by the fixed subset (which can be located very far from both singlets), so that the relevant configurations remain effectively linearly independent.}. 
This is in sharp contrast to the power-law decay of dimer--dimer correlations in the statistical ensemble of all states in the Klein subspace~\cite{Nussinov2007a}. This observation has an interesting implication -- a novel order-by-disorder mechanism outlined in the next section.
\begin{figure}[thb]
    \centering
    \includegraphics[width=0.5\textwidth]{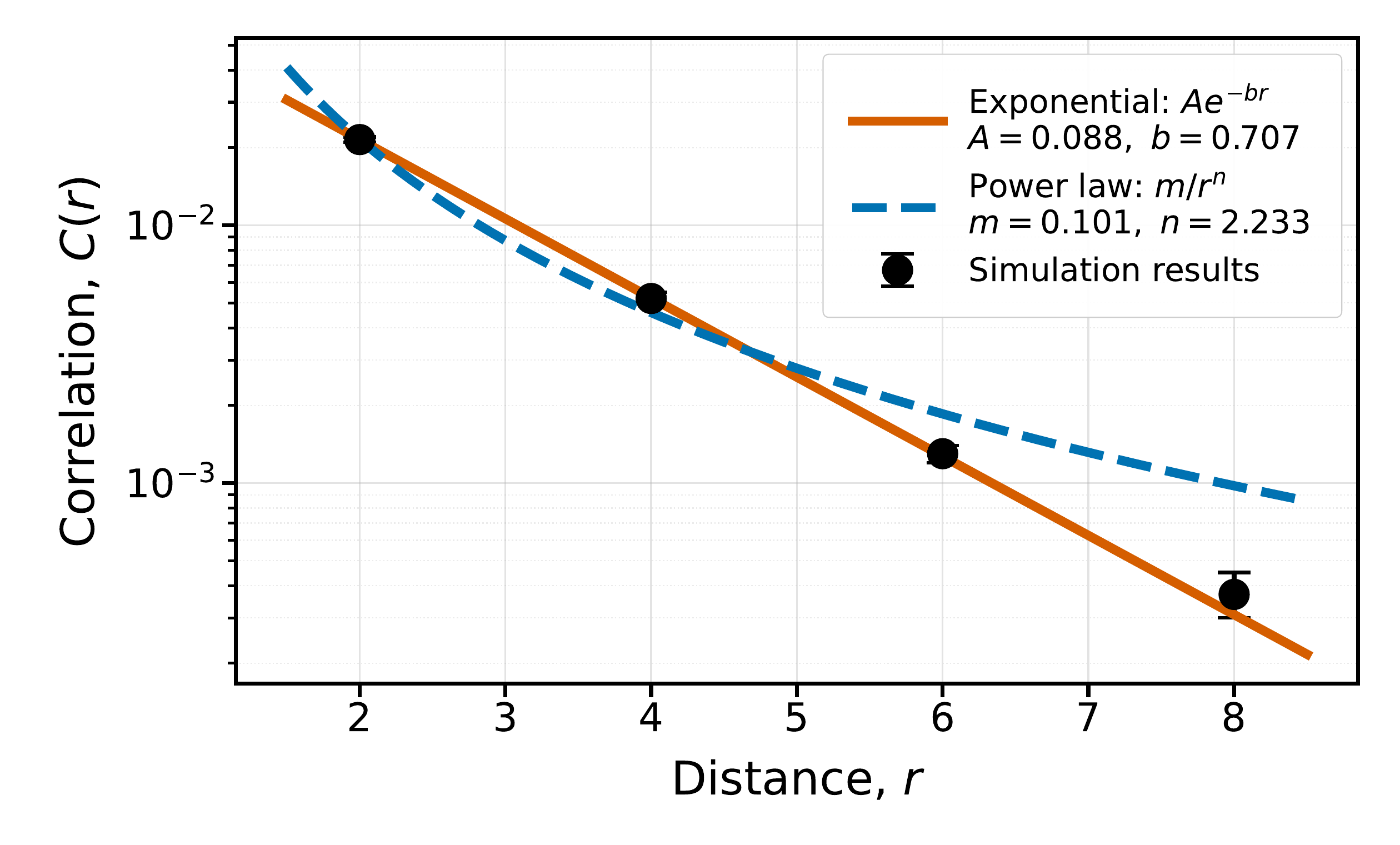}
    \caption{Singlet--singlet correlation function between two horizontal singlets in the same vertical column as a function of distance,
    for a $24 \times 24$ lattice. The data are consistent with exponential
    decay, well described by $e^{-r/1.41}$ law; the best power-law fit, $r^{-2.23}$, is also shown for comparison.}
    \label{fig:correlations}
\end{figure}

\subsubsection{Order by Disorder driven by thermal decoherence}
\label{sec:ObDO}

Again, consider the RVB state that is the ground state of the Hamiltonian in Eq.~(\ref{eq:full_Hamiltonian}).
This Hamiltonian consists of two contributions: the Klein-type Hamiltonian with coupling constant $K$, whose ground-state manifold is an ensemble of VB states in the Klein subspace, and an additional term with coupling constant $J$, which endows the VB states with quantum dynamics and leads to the RVB ground state. In general, it is assumed that $J\ll K$.

Let us now focus on the temperature regime with $K\gg T\gg J$. Intuitively, thermal fluctuations are insufficient to excite states outside the Klein subspace, so the physically relevant states are still the VB states. However, the temperature may be sufficiently high to decohere such VB states in the RVB ground state. In essence, this mechanism is quite similar to that considered in Ref.~[\onlinecite{Castelnovo2007}] in the context of the toric code. In terms of the density matrix, one would intuitively expect the vanishing of off-diagonal matrix elements and turning the pure state into a mixed state consisting of different VB states taken with the same statistical weights. There is, however, a caveat: the VB states do not form an orthogonal basis! Nevertheless, this does not present a difficulty; the density matrix is still well defined in this basis. In particular, the trace of the diagonal density matrix is equal to one if individual basis states $|\mathcal{D}_i\rangle$ are normalized. In terms of the (real, symmetric) overlap matrix $G_{ij}=\langle\mathcal{D}_i|\mathcal{D}_j\rangle$, this means $G_{ii}=1$ and $\text{Tr}\,G = \mathcal{N}$, where $\mathcal{N}$ is the dimensionality of the VB basis. Then, using orthonormal states $\lvert \phi_i\rangle = \sum_j\left(G^{-1/2}\right)_{ij}\lvert \mathcal{D}_j\rangle$ we obtain
\begin{multline}
    \text{Tr} \,\rho = \sum_j\langle\phi_j\left|\frac{1}{\mathcal{N}}\sum_i|\mathcal{D}_i\rangle\langle \mathcal{D}_i|\right| \phi_j\rangle\\
    =\frac{1}{\mathcal{N}}\sum_{k,i,l} (G^{-1})_{kl}\langle \mathcal{D}_k|\mathcal{D}_i\rangle\langle \mathcal{D}_i|\mathcal{D}_l\rangle  =\frac{1}{\mathcal{N}}\text{Tr}\,G=1.
\end{multline}

Since we are dealing with the diagonal density matrix, for any operator $\hat{O}$, 
$
\langle\hat{O}\rangle=\text{Tr}(\rho\hat{O})={\sum O_{ii}}/{\mathcal{N}},
$
where $O_{ii}=\langle{\mathcal{D}}_i|\hat{O}|\mathcal{D}_i\rangle$. In particular, this implies that the singlet--singlet correlation function in such a mixed state is given by the dimer--dimer correlation function (up to a trivial factor of $9/16$~\cite{Tang2011a}), which is, in turn, governed by the vertex--vertex correlations in the classical six-vertex model that fall off as $1/r^2$~\cite{Youngblood1980,Nussinov2007a}.

However, this conclusion is hinged on the linear independence of the VB states $|\mathcal{D}\rangle$; in particular, we have assumed the existence of $G^{-1}$. Proofs of linear independence of such states exist only for a limited set of lattices~\cite{Chayes1989c,Seidel2009,Wildeboer2011}). While we do not have a proof of such linear independence for the entire Klein subspace considered here (and worse yet, our numerical checks for small lattices resulted in a rank of $G$ slightly smaller than the size of the Klein subspace), we have proved that the singlet covering states \emph{are} linearly independent on an arbitrarily large lattice with a fixed singlet covering of a finite number (as low as 4) of adjacent checkerboard squares (see Supplementary Materials for details). These fixed singlets can be considered as a sort of ``boundary conditions'', which, in principle, leads to another small caveat: The known results for the correlations in the six-vertex model (and, by extension, for the dimer--dimer correlations relevant here) have been derived for the cases not subjected to those ``boundary conditions''. We argue that this should not present a problem for our conclusions: the effects of those fixed singlets should themselves fall off as $1/r^2$ and since they can be removed arbitrarily far from the singlets whose correlation we are interested in, the effect of these ``boundary conditions'' can be made arbitrarily small. In particular, it cannot reduce the expected power-law correlations between the singlets of interest to a shorter-range, exponential decay. Note that the six-vertex model has a line in its phase diagram where it can be mapped to free fermions and where such statements can be made more rigorous due to Wick's theorem. Although this line does not pass through the point of our interest where all vertex weights are the same, it lies in the same phase, providing us with further confidence.  

\subsubsection{Conclusions}
We have constructed a local SU(2) symmetric spin-1/2 Hamiltonian on the checkerboard lattice whose ground state is an RVB spin liquid. Although we did not address the question of the spin gap, we have presented numerical evidence of the exponential decay of singlet-singlet correlations, which is consistent with the general expectation of such decay on non-bipartite lattices. However, the nature of the constraints imposed on the VB manifold by the leading Klein-type terms in the Hamiltonian (which form the starting point of our construction) gives rise to a novel mechanism of thermal order by disorder. Specifically, as thermal decoherence destroys the phase coherence between different VB states in the ground state RVB superposition, the density matrix becomes diagonal in the VB basis. Surprisingly, this leads to a qualitative \emph{increase} in singlet--singlet correlations, whose decay becomes power-law instead of exponential. In other words, the correlations increase due to the suppression of destructive interference between different VB states participating in the RVB ground state on the frustrated checkerboard lattice.  The emergence of a quasi-long range order serves as an example of a novel route for the onset of thermal order by disorder: instead of the entropic effects of thermal fluctuations, the new mechanism relies on the suppression of destructive interference behind the quantum disorder at zero temperature.

\begin{acknowledgments}
The authors thank Hong Yao for sharing details of his earlier work. Our numerics were performed using HPCC computer clusters and data storage resources, which were funded by grants from NSF (MRI-2215705, MRI-1429826) and NIH (1S10OD016290-01A1).
\end{acknowledgments}

\newpage
\onecolumngrid
\newpage
\section*{Supplementary Materials}
\setcounter{subsection}{0} 
\subsection{Proof of the RK Parent Hamiltonian Structure}

Consider a $4\times 4$ cyclic exchange operator $\hat{P}$ satisfying
\begin{equation}
\hat{P}^4 = I.
\end{equation}

Its eigenvalues are the quartic roots of unity
\begin{equation}
\lambda \in \{1,\,i,\,-1,\,-i\},
\end{equation}
with the corresponding eigenvectors
\begin{align}
(1,1,1,1), \quad
(1,-i,-1,i), \quad
(1,-1,1,-1), \quad
(1,i,-1,-i).
\end{align}

Define the operator
\begin{equation}
\hat{\mathcal{H}} = \left(I - \hat{P} + \hat{P}^2 - \hat{P}^{-1}\right).
\end{equation}

Since $\hat{P}^{-1}=\hat{P}^3$, the eigenvalues of $\hat{\mathcal{H}}$ obey
\begin{equation}
\mathcal{H}(\lambda)
=
\left(1 - \lambda + \lambda^2 - \lambda^{-1}\right).
\end{equation}

Substituting $\lambda = 1,i,-1,-i$ yields
\begin{equation}
{\mathcal{H}}(\lambda) = \{0,\,0,\,2,\,0\}.
\end{equation}

Hence, $\hat{\mathcal{H}}$ is a positive semi-definite operator with a three-dimensional kernel.

\medskip
When this operator acts on spin indices of $\mathbf{S}_1$, $\mathbf{S}_2$, $\mathbf{S}_3$ and $\mathbf{S}_4$ in Fig.~\ref{fig:numbering} of the main text, we should additionally require that the resulting state remains in the Klein subspace.
This restricts the allowed basis vectors to
\begin{equation}
(1,0,0,0),
\qquad
(0,1,0,0).
\end{equation}
which are associated with the two states of a flippable plaquette.

The most general vector in the kernel of $\hat{\mathcal{H}}$ can be expressed as a linear combination of the three zero-eigenvalue eigenvectors,
\begin{equation}
\psi =
\left(
c_1 + c_2 + c_3,\;
c_1 - i c_2 + i c_3,\;
c_1 - c_2 - c_3,\;
c_1 + i c_2 - i c_3
\right).
\end{equation}

Imposing the physical constraint that only the first two components are nonzero gives
\begin{align}
c_1 - c_2 - c_3 &= 0, \\
c_1 + i c_2 - i c_3 &= 0.
\end{align}

Solving these equations yields
\begin{equation}
\psi = (2c,\,2c,\,0,\,0).
\end{equation}

Therefore, within the projected Hilbert space (the Klein subspace), the unique ground state is proportional to
\begin{equation}
(1,0,0,0) + (0,1,0,0).
\end{equation}
\newpage

\newpage
\subsection{Linear independence of the checkerboard Valence Bond states with some fixed singlets}

In this section, we prove that fixing singlets on a finite number of plaquettes (henceforth referred to as a patch) guarantees linear independence of the  singlet coverings of the entire lattice which contain those fixed singlets.  While the smallest such patch consists of 4 plaquettes connected by their corners and forming a $4\times 1$ strip, we begin by showing the proof for a \(3 \times 3\) patch shown in Fig.~\ref{fig:fixedsinglets}, as it is simpler. The \(4 \times 1\) case follows the same steps; however, it requires one additional step to ensure that the ``patch expansion'' procedure is always possible, so it is discussed at the end. Before proceeding with the proof, let us clarify the terminology used in this section. The patches considered here are formed by connecting corner-sharing crossed plaquettes. They are referred to as ``rectangular'' whenever the outline of the overall shape of a patch is a rectangle (e.g., as seen in Fig.~\ref{fig:fixedsinglets}) even though the plaquettes forming it are oriented at $45^\circ$ to the sides of the rectangle containing the patch.

The method of the proof is a slight modification of that in \cite{Seidel2009} and is based on construction projector operators that project onto a given singlet covering and project the rest of them out. When describing the singlet configurations, we will use the language of the six-vertex model owing to the one-to-one correspondence between six-vertex and dimer configurations outlined in the main text. 

\begin{figure} [h]
    \centering
    \includegraphics[width=0.5\linewidth]{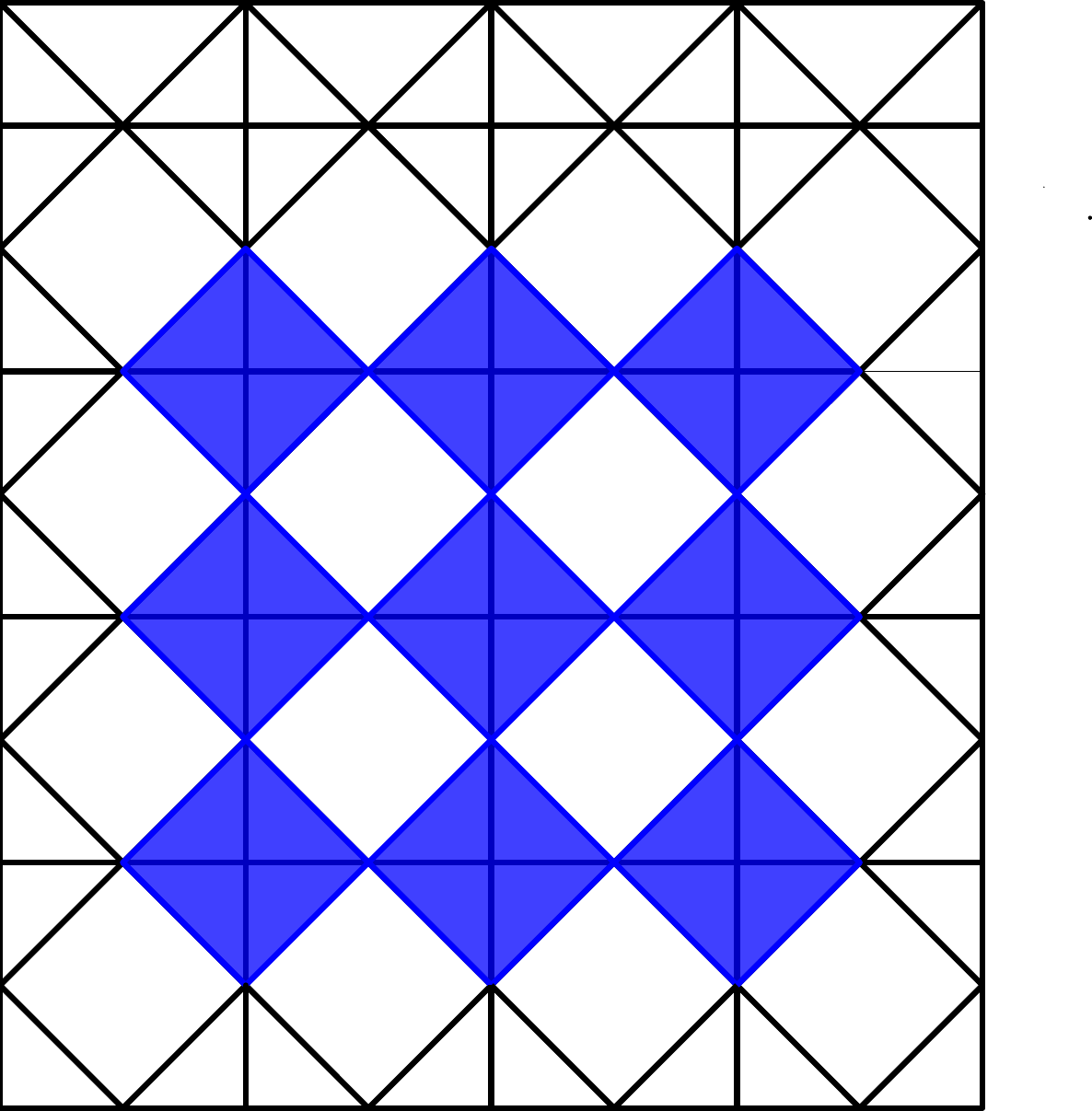}
    \caption{Plaquettes in which singlet configurations are fixed.}
    \label{fig:fixedsinglets}
\end{figure}

\begin{lemma} 
\label{lemma1}
For any rectangular patch formed of crossed plaquettes there exists at least one side on which the number of
outgoing arrows of the six-vertex model is greater than or equal to the number of incoming arrows. \\
\end{lemma}

\begin{proof} Assume the opposite. Then the number of incoming arrows
is strictly greater than the number of outgoing arrows for each side of the rectangle. Hence,
more arrows would enter the closed boundary of the rectangle than leave it, which is impossible for the six-vertex model since all of its vertices are divergence-free.
\end{proof} 

\begin{nres}
\label{nres1}
Consider a $3\times 1$ patch made of three plaquettes as shown in Fig.~\ref{fig:Boundary}. If the boundary condition defined by arrows through the bottom dashed rectangle is fixed so that at least two arrows are incoming, then the corresponding configurations are linearly independent.
\end{nres}
\begin{figure} [h]
    \centering
    \includegraphics[width=0.35\linewidth]{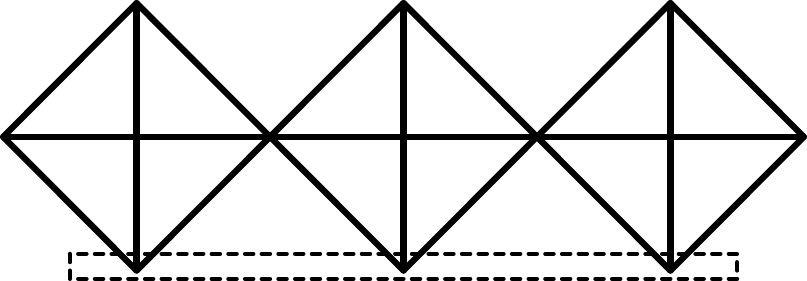}
    \caption{A strip consisting of three plaquettes. The bottom boundary is indicated by the dashed rectangle.}
    \label{fig:Boundary}
\end{figure}

\begin{nres}
\label{nres2}
If the directions of \emph{two} arrows are fixed within a single plaquette, then the singlet configurations on that plaquette that are consistent with the two fixed arrows are linearly independent.
\end{nres}

\begin{proof}[\textit{Proof of the main claim for a \(3 \times 3\) patch:}]
Fix an arbitrary singlet configuration on a \(3 \times 3\) patch. Our goal is to construct a projector onto a state with an arbitrary singlet configuration on the remaining plaquettes consistent with the fixed configuration on the \(3 \times 3\) patch, which in turn guarantees linear independence of all such states~\cite{Seidel2009}.

It follows from Lemma~\ref{lemma1} that at least one side of the \(3\times3\) patch has at least as many outgoing
arrows as incoming arrows. Without loss of generality, let this side to be the top side in  Fig.~\ref{fig:fixedsinglets}. Then the three plaquettes adjacent to this side, shown in red in Fig.~\ref{fig:placeholder}, have at least two incoming arrows entering its bottom boundary. Therefore, Numerical Result~\ref{nres1} ensures that the corresponding configurations are linearly independent, and a local projector can be constructed for these three plaquettes~\cite{Seidel2009}.

\begin{figure} [h]
    \centering
    \includegraphics[width=0.5\linewidth]{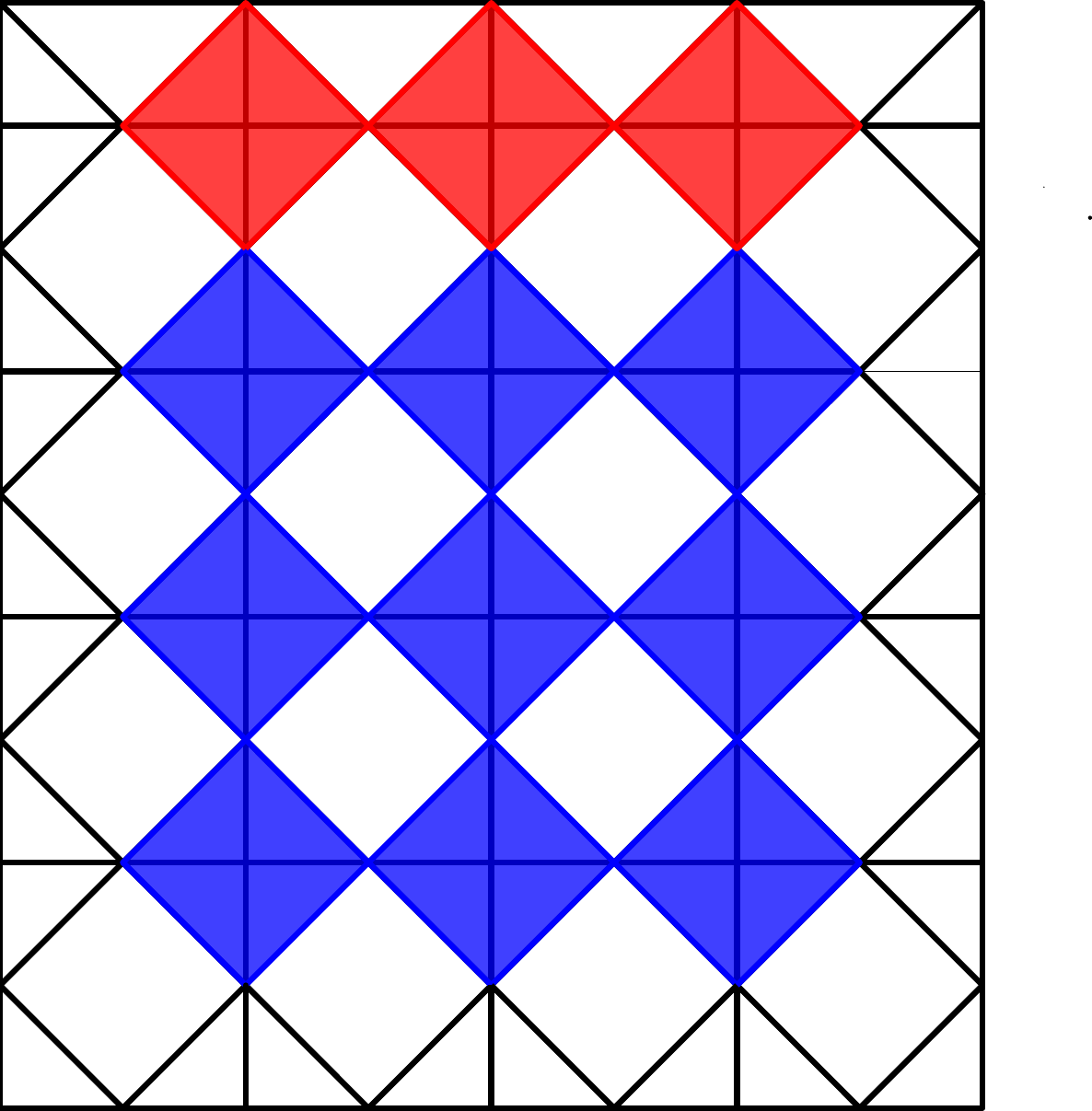}
    \caption{Expansion procedure for the $3\times3$  patch (blue-shaded plaquettes). If the top boundary of the original patch has more outgoing arrows than incoming arrows, the adjacent red plaquettes have more incoming arrows than outgoing arrows and thus can be incorporated into an expanded patch.}
    \label{fig:placeholder}
\end{figure}

Now repeat this construction by expanding the patch one side at a time and constructing the corresponding projector operator at each step. Suppose that the current rectangular patch has a side of length $L>3$ with at least as many outgoing arrows as incoming arrows  (see Fig.~\ref{fig:expand}). There must exist a set of three neighboring boundary plaquettes with at least two outgoing arrows (For demonstration, without loss of generality, let them be blue-shaded plaquettes in Fig.~\ref{fig:expand}). Now consider three new plaquettes immediately adjacent to them  (red-shaded plaquettes in Fig.~\ref{fig:expand}). These red-shaded plaquettes have at least two arrows incoming from below. It  then follows from Numerical Result~\ref{nres1} that they are linearly independent, and thus constructing a local projector for this additional three-plaquette strip is possible.

\begin{figure} [h]
    \centering
    \includegraphics[width=0.45\linewidth]{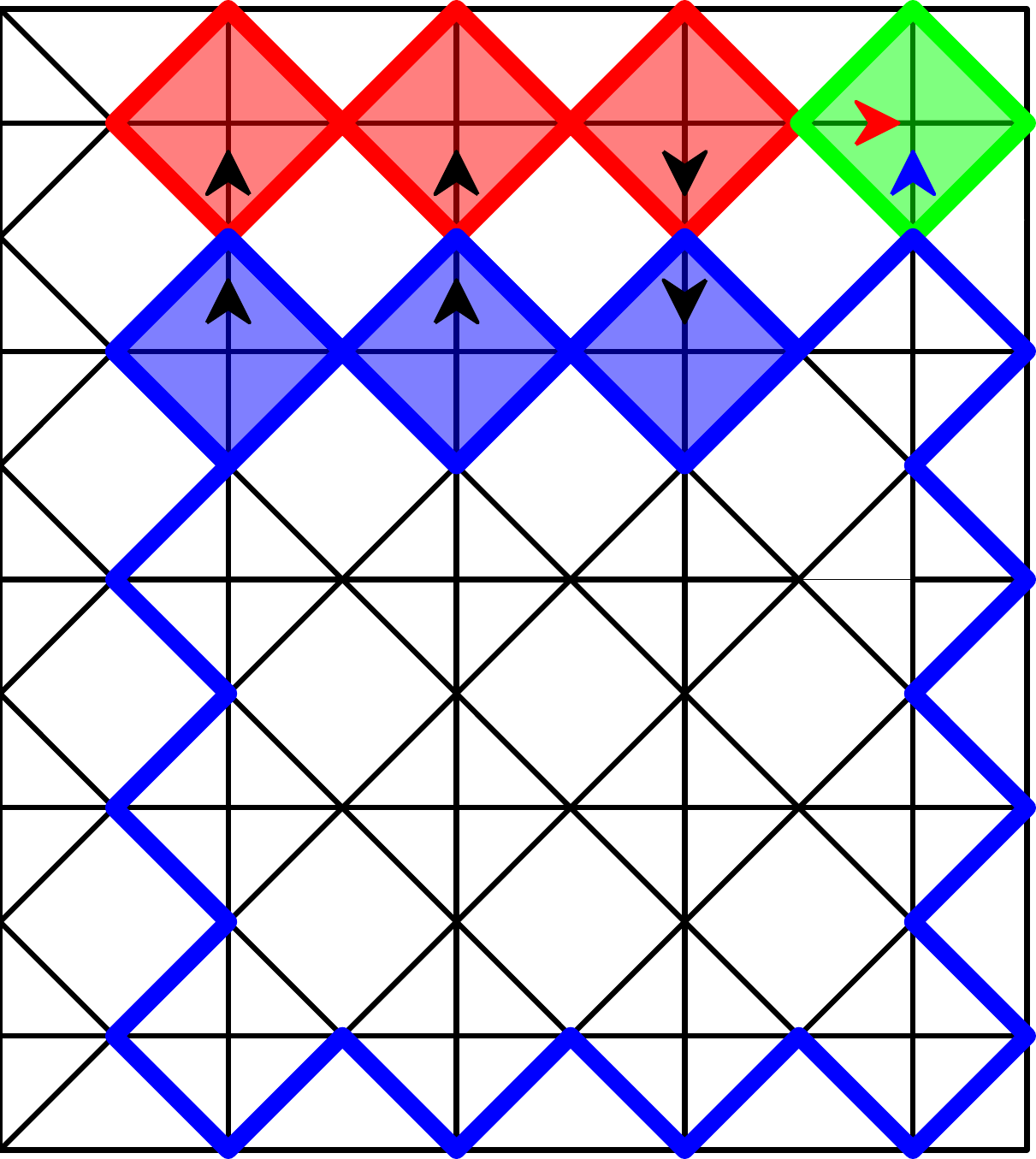}
    \caption{Shaded blue plaquettes are the starting point of the patch expansion: they provide the necessary boundary condition which allows for the addition of red-shaded plaquettes. The arrows represent just one example of such a boundary condition; the  only constraint on these arrows is that there are more arrows oriented from  blue to red plaquettes than than those with the opposite orientation. The last step is adding a green-shaded plaquette  to restore the expanded patch to a rectangle form. In principle, there can be any number of green-shaded plaquettes that allow for the sideway expansion of the original red-shaded ``seed'', one at a time and in both directions. Again, the arrows shown here are just an example; the key is that they are fixed by the patch below and the growing expansion seed on the left.}
    \label{fig:expand}
\end{figure}

Having applied the projector to fix the state of the strip, we can use it as a ``seed'' to grow the strip horizontally (generically, in both directions) until the side of the newly expanded rectangular patch is complete. The procedure is as follows: the plaquettes horizontally adjacent to the seed (e.g. the green-shaded plaquette in Fig.~\ref{fig:expand} ) have two fixed arrows - one from the already projected blue patch below (one possible arrow orientation is shown as a blue arrow in Fig.~\ref{fig:expand}) and the other from the newly projected seed (a red arrow incident from the red strip in Fig.~\ref{fig:expand}). Therefore, from Numerical Result~\ref{nres2}, all possible configurations on this green plaquette are also linearly independent, and the corresponding projector operators can be constructed and applied. If necessary, repeat this step for subsequent horizontally adjacent plaquettes until the side of the expanded rectangular patch is complete (in Fig.~\ref{fig:expand} it is not necessary, since fixing the state of just one green plaquette completes the side of the expanded rectangle).

Thus, the patch can always be expanded until the entire lattice can be covered by repeating this procedure. Here we assume periodic boundary conditions in order to ensure that the procedure cannot get stuck if a ``good'' side of the patch hits the boundary while the ``bad'' sides are not amenable to the aforementioned expansion procedure. The final projector operator is then given by the product of the local projector operators applied in the order determined by the patch-expansion procedure. Finally, as mentioned above and outlined in Ref.~\cite{Seidel2009}, the existence of such a projector operator for each configuration automatically implies linear independence of the configurations. 
\end{proof}

\begin{proof}[Proof for a \(4 \times 1\)  patch:]
The logic of the proof is the same if the starting point is a \(4 \times 1\) patch; however, one must show that the patch can always be expanded. 

Initially, the boundary points $(A)$ and $(B)$ may have incoming and one outgoing (one-in/one-out), two-in, or two-out arrows (shown in Fig.~\ref{fig:4X1}). It is straightforward to verify that, in every case,  one of the two dashed horizontal boundaries in Fig.~\ref{fig:4X1} contains at least as many outgoing arrows as incoming ones. The least obvious case is the two-out configuration; however, even in this case, one of the two dashed horizontal boundaries must have a one-in/three-out arrow configuration, while the other boundary must have a two-in/two-out configuration, so that the closed boundary has an equal number of incoming and outgoing arrows. Consequently, the patch can be expanded in the vertical direction across the two-in/two-out boundary, resulting in a \(4\times2\) patch (see Fig.~\ref{fig:4X2}). 
\begin{figure} [h]
    \centering
    \includegraphics[width=0.5\linewidth]{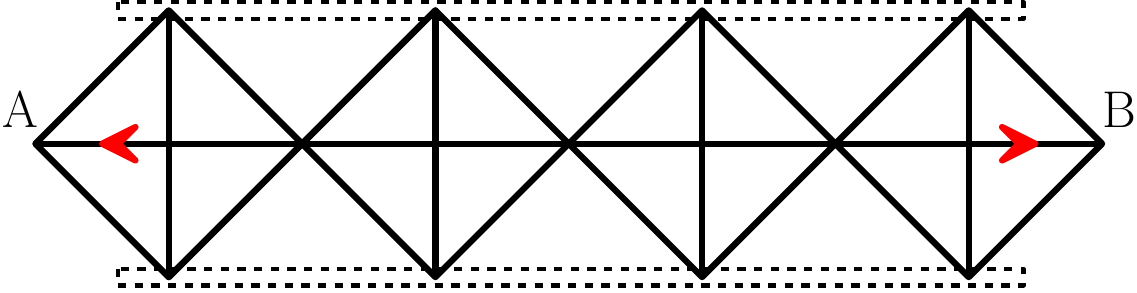}
    \caption{A \(4\times1\)  rectangle with sites $(A)$ and $(B)$ corresponding to a horizontal boundary with two outgoing sideway arrows.}
    \label{fig:4X1}
\end{figure}
\\ 
\begin{figure} [h]
    \centering
    \includegraphics[width=0.5\linewidth]{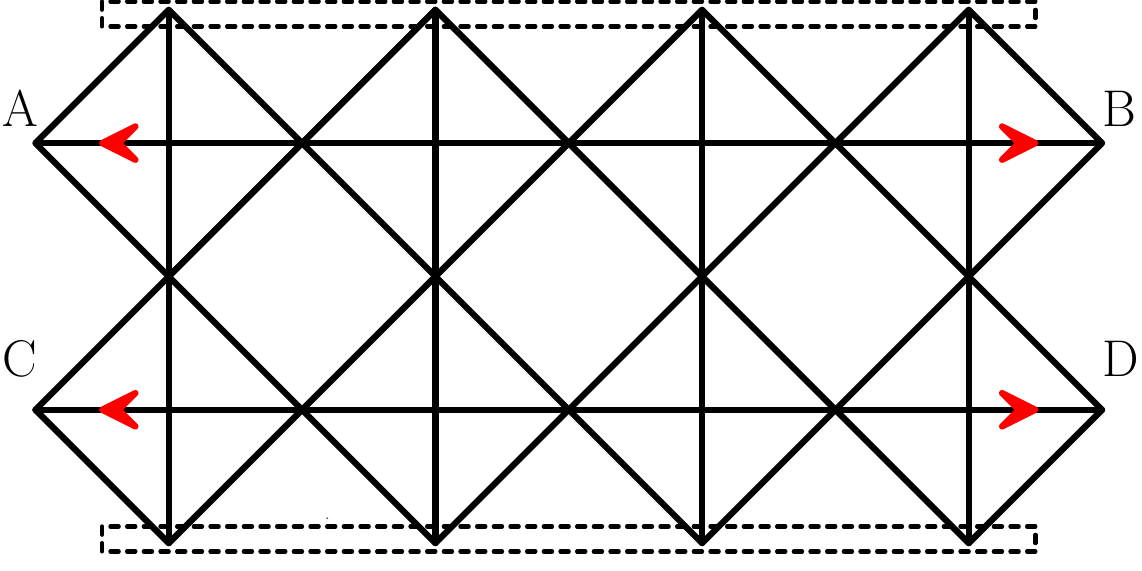}
    \caption{\(4\times2\) rectangle with four outgoing arrows on the vertical boundary.}
    \label{fig:4X2}
\end{figure}

Before continuing, we present one additional numerical result:
\begin{nres}
\label{nres3}
Consider a two-plaquette strip such as the one shown in Fig.~\ref{fig:2-plaquettes}. If the boundary condition at the dashed boundary is fixed so that exactly two arrows are incoming, then the corresponding configurations are linearly independent.
\end{nres}
\begin{figure} [h]
    \centering
    \includegraphics[width=0.15\linewidth]{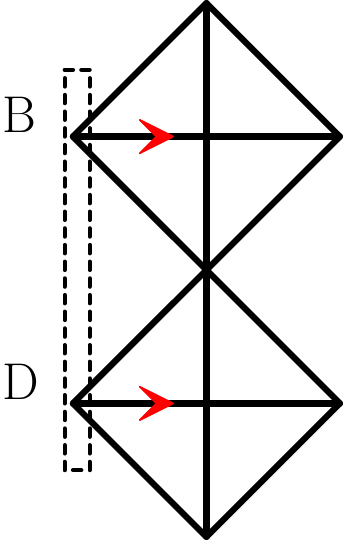}
    \caption{A two-plaquette strip with two incoming arrows at the dashed vertical $(B,D)$ boundary.}
    \label{fig:2-plaquettes}
\end{figure}

Now consider a \(4\times2\) patch. If neither of the vertical boundaries (sides $(A,C)$ and $(B,D)$ in Fig.~\ref{fig:4X2}) has two outgoing arrows, then at least one of the two horizontal dashed boundaries contains at least as many outgoing arrows as incoming ones. Therefore, the patch can be expanded in the vertical direction, resulting in a new rectangular patch with the length of the sides of at least 3 (the expansion procedure is the same as the aforementioned procedure for a \(3\times3\) patch). Consequently, from that point, the proof given for the previous \(3\times3\) case applies.

If, on the other hand, at least one vertical boundary has two outgoing arrows, then Numerical Result~\ref{nres3} ensures that the patch can be expanded in the horizontal direction.  We continue horizontal expansion until neither vertical boundary has two outgoing arrows. At that point, vertical expansion becomes possible, reducing the problem to the \(3\times3\) case proved previously. 

Alternatively, if after every step of the horizontal expansion results in two outgoing arrows at at least one of the vertical boundaries, then one must rely on the periodic boundary conditions in the horizontal direction. Upon completing an $L_x\times 2$ horizontal belt (with $L_x$ being the horizontal circumference), we can employ Lemma~\ref{lemma1} which guarantees that the expansion in the vertical direction is now possible. 
\end{proof}


\end{document}